\documentclass{article}

\usepackage{arxiv}

\usepackage{amsmath,amssymb,amsfonts}
\usepackage[utf8]{inputenc} 
\usepackage[T1]{fontenc}    
\usepackage{hyperref}       
\usepackage{url}            
\usepackage{booktabs}       
\usepackage{amsfonts}       
\usepackage{nicefrac}       
\usepackage{microtype}      
\usepackage{lipsum}
\usepackage{amsthm}
\usepackage{graphicx}
\usepackage{textcomp}

\newtheorem{definition}{Definition}[section]
\newtheorem{theorem}{Theorem}[section]

\newtheorem{lemma}[theorem]{Lemma}
\newtheorem{example}{Example}[section]
\newtheorem{remark}{Remark}[section]
\newtheorem{prop}{Proposition}

\title{A note on an OACF-preserving operation based on Parker's Transformation\thanks{This paper was presented in part at the 2019 Ninth International
Workshop on Signal Design and its Applications in Communications (IWSDA), Dongguan, China, Oct. 20-24, 2019. However, the authors misunderstood that the decimation operation preserves the odd-periodic autocorrelation function (OACF) values. The present extended version clarifies this point and gives a complete description of three OACF-preserving opertaions, which are exactly the counterparts of the three operations preserving periodic autocorrelation function (PACF) values.}}

\author{
  Geyang Wang \\
  Department of Computer Science and Engineering\\
  Southern University of Science and Technology\\
  Shenzhen, Guangdong 518055, China \\
 \texttt{11510050@mail.sustech.edu.cn} \\
   \And
  Qi Wang\\
  Department of Computer Science and Engineering\\
  Southern University of Science and Technology\\
  Shenzhen, Guangdong 518055, China\\
  \texttt{wangqi@sustech.edu.cn} \\
}

\begin{document}
\maketitle

\begin{abstract}
Binary sequences with low odd-periodic correlation magnitudes have found important applications in communication systems. It is well known that the nega-cyclic shift and negation preserve the odd-periodic autocorrelation function (OACF) values in general. In this paper, we define a new operation based on Parker's transformation, which also preserves the OACF values of binary sequences. This enables us to classify Parker's 16 cases into 8 ones, and may possibly further allow to classify all constructions based on Parker's transformation.
\end{abstract}

\section{Introduction}

A periodic sequence is called \textbf{binary} if every element of the sequence is from $\mathbb{Z}_2 = \{0,1\}$. For two binary sequences $\mathbf{a} = (a(i))$ and $\mathbf{b} = (b(i))$ of period $N$, where the sequence indices $i$ are taken modulo $N$. The {\em odd-periodic correlation function} of $\mathbf{a}$ and $\mathbf{b}$ at the shift $\tau$ is defined as 

\begin{equation}\label{def: odd periodic correlation}
  \mathcal{\widehat{R}}_{\mathbf{a},\mathbf{b}}(\tau) = \sum_{i=0}^{N-1}(-1)^{a(i) - b(i + \tau) + \left \lfloor \frac{i+\tau}{N} \right \rfloor}, \text{ } 0\le \tau < N
\end{equation}
where $\lfloor \frac{i+\tau}{N} \rfloor$ denotes the integer part of $\frac{i+\tau}{N}$. 
When the two sequences $\mathbf{a}$ and $\mathbf{b}$ are identical, the odd-periodic correlation function $\mathcal{\widehat{R}}_{\mathbf{a},\mathbf{b}}(\tau)$ is called {\em negaperiodoc autocorrelation function} (NACF) or {\em odd-periodic autocorrelation function} (OACF) of $\mathbf{a}$, and simply denoted by $\mathcal{\widehat{R}_{\mathbf{a}}}(\tau)$. The (even)-{\em periodic correlation function} of $\mathbf{a}$ and $\mathbf{b}$ at the shift $\tau$ is defined as 
\begin{equation}\label{def: periodic correlation}
   \mathcal{R}_{\mathbf{a},\mathbf{b}}(\tau) = \sum_{i=0}^{N-1}(-1)^{a(i) - b(i + \tau)},\text{ }  0\le \tau < N.
\end{equation}
The (even)-{\em periodic autocorrelation function} (PACF) of $\mathbf{a}$ is $\mathcal{R}_{\mathbf{a},\mathbf{a}}(\tau)$ and denoted by $\mathcal{R}_{\mathbf{a}}(\tau)$. Binary sequences with low PACF magnitudes have been extensively investigated, and are closely related to combinatorial objects (see, for example, a survey \cite{cai2009binary}). 
Very recently, there was progress on binary sequences with low OACF magnitudes (see \cite{yang2015binary},\cite{li2019three},\cite{yang2017generic}). The lower bound of OACF magnitudes was given by Pott \cite{pott2008two}. In \cite{yang2015binary}, the authors give an alternative proof for the lower bound.

\begin{prop}\cite[Theorem 2]{yang2015binary}
  Let $\mathbf{a}$ be a binary sequence of period $N$, then
  \begin{equation} \label{equ : lower bound of odd periodic autocorrelation}
    \widehat{\theta}(\mathbf{a}) \ge 
    \begin{cases}
      1 , & \text{if } N \equiv 1 \mod 2, \\
      2 , & \text{if } N \equiv 0 \mod 2.
    \end{cases}
  \end{equation}
  where $\widehat{\theta}(\mathbf{a}) = \max_{0 < \tau < N}\left | \mathcal{\widehat{R}_{\mathbf{a}}(\tau)} \right |$.
\end{prop}

The periodic binary sequence $\mathbf{a}$ is called {\em odd-optimal} if the lower bound in Eq. \eqref{equ : lower bound of odd periodic autocorrelation} is met. 

Binary sequences with low OACF magnitudes can be derived from those with low PACF magnitudes by Parker's transformation \cite{parker2001even}. The Parker's transformation is defined as follows: Let $\mathbf{s}$ be a binary sequence with length $N$. Define $\mathbf{u} = \mathbf{s} || (\mathbf{s} \oplus 1) = [s_0, s_1, s_2, \dots, s_{N-1}, s_0 + 1, s_1 + 1, s_2 + 1, \dots, s_{N-1} + 1]$, where $||$ denotes the concatenation. Notice that for $0\le \tau < N$, we have

 \begin{eqnarray*}
     \mathcal{R}_{\mathbf{u}}(\tau) &=&  \sum_{k=0}^{2N-1} (-1)^{\mathbf{u}(k) - \mathbf{u}(k+\tau)}  \\
        &=& \sum_{k=0}^{N-\tau-1} (-1)^{\mathbf{u}(k) - \mathbf{u}(k+\tau)} + \sum_{k=N-\tau}^{N-1} (-1)^{\mathbf{u}(k) - \mathbf{u}(k+\tau)} \\
         &&+ \sum_{k=N}^{2N-\tau-1}(-1)^{\mathbf{u}(k) - \mathbf{u}(k+\tau)}  + \sum_{k=2N-\tau}^{2N-1} (-1)^{\mathbf{u}(k) - \mathbf{u}(k+\tau)}\\
        &=& \sum_{k=0}^{N-\tau-1} (-1)^{\mathbf{s}(k) - \mathbf{s}(k+\tau)} +
        \sum_{k=N-\tau}^{N-1} (-1)^{\mathbf{s}(k) - \mathbf{s}(k+\tau)-1} \\
        &&+  \sum_{k=N}^{2N-\tau-1}(-1)^{\mathbf{s}(k) + 1 - \mathbf{s}(k+\tau) - 1} + \sum_{k=2N-\tau}^{2N-1} (-1)^{\mathbf{s}(k) + 1 - \mathbf{s}(k+\tau)}\\
        &=& 2(\sum_{k=0}^{N-\tau-1} (-1)^{\mathbf{s}(k) - \mathbf{s}(k+\tau)} +
        \sum_{k=N-\tau}^{N-1} (-1)^{\mathbf{s}(k) - \mathbf{s}(k+\tau)-1})\\
        &=& 2\mathcal{\widehat{R}}_{\mathbf{s}}(\tau).
 \end{eqnarray*}

The OACF values of the sequence $\mathbf{s}$ can thereby be computed by the PACF values of the sequence $\mathbf{u}$ via the following equation.

\begin{equation} \label{equ : Parker transformation}
  \mathcal{\widehat{R}}_{\mathbf{s}}(\tau) = \frac{\mathcal{R}_{\mathbf{u}}(\tau)}{2},
\end{equation}
where $0 \le \tau < N$. 
Thus, instead of constructing binary sequence $\mathbf{s}$ with low OACF magnitudes directly, Parker \cite{parker2001even}, Li and Yang \cite{li2019three} constructed sequences with low OACF magnitudes by first constructing sequences with low PACF magnitudes of the form  $\mathbf{u} = \mathbf{s} || (\mathbf{s} \oplus 1)$, and it follows that $\mathbf{s}$ has low OACF magnitudes.

In \cite{parker2001even}, Parker mentioned that negation, cyclic shift and decimation on $\mathbf{u}$ can be translated into operations on $\mathbf{s}$, and the translated operations preserve the distribution of OACF magnitudes. The transformation $\delta$ on periodic binary sequence is called OACF-preserving if $\mathbf{s}$ and $\delta(\mathbf{s'})$ have the identical distribution of OACF values, where $\mathbf{s},\mathbf{s'}$ are periodic sequences with same period.
The new decimation defined in \cite{wang2019oacf} was proved to be OACF-preserving, based on Parker's transformation. However, the authors misunderstood the three OACF-equivalent operations mentioned by Parker in \cite{parker2001even}, and claimed this new decimation (here called {\em nega-decimation}) constitutes a fourth OACF-preserving operation. In the present paper, we restate the main results in \cite{wang2019oacf}, and make it clear that there are three OACF-preserving operations (negation, nega-cyclic shift, nega-decimation). They are exactly the counterparts of the three PACF-preserving operations, i.e., negation, cyclic shift, and decimation. Furthermore, we classify Parker's 16 cases \cite{parker2001even} into 8 ones, based on the equivalence relation defined on the three OACF-preserving operations. 

The remainder of the present paper is organized as follows. In Section \ref{sec: Section 2}, we define nega-decimation and prove that nega-decimation preserves the OACF values distribution.  Moreover, we give an example to show that decimation is not OACF-preserving. In Section \ref{Sec: section 3}, by applying nega-decimation, we are able to classify the sequences conjectured by Parker and confirmed by Li and Yang \cite{yang2017generic} into fewer cases. Finally, Section \ref{Sec: Conclusion} concludes this paper. 

\section{Nega-decimation} \label{sec: Section 2}
In this section, we first review the four well-known operations for periodic binary sequences: negation, cyclic shift, nega-cyclic shift and decimation. Next, we show that decimation is in general not OACF-preserving by giving an example. We then define the nega-decimation and prove that nega-decimation is OACF-preserving. Furthermore, by giving an example, it is shown that the nega-decimation operation cannot be obtained by any combination of negation and nega-cyclic shift in general. 

We review the definition of negation, cyclic shift, decimation in the following definitions.
\begin{definition} \label{def: negation}
  Let $\mathbf{a}$ be a binary sequence of period $N$. The {\em negation} on $\mathbf{a}$ is defined as 
  $$
  [\mathbf{a}(0) + 1, \mathbf{a}(1) + 1,\dots, \mathbf{a}(N-1) + 1].
  $$ and denoted by  $\mathbf{a}\oplus 1 $ or $\bar{\mathbf{a}}$.
\end{definition}

\begin{definition} \label{def : cyclic shift}
  Let $\mathbf{a}$ be a binary sequence of period $N$. The {\em cyclic shift} $\tau$, denoted by $L^\tau$, is defined as
  $$
    L^{\tau}(\mathbf{a}) = [\mathbf{a}(\tau),\mathbf{a}(\tau + 1), \dots, \mathbf{a}(N-1), \mathbf{a}(0),\dots,\mathbf{a}(\tau-1)],
  $$
  where $0\le \tau< N$.
\end{definition}

\begin{definition} \label{def : decimation}
  Let $\mathbf{a}$ be a binary sequence of period $N$. For any integer $d$ with $\gcd(d,N)=1$, the {\em $d$-decimation} on $\mathbf{a}$ is a mapping that maps $\mathbf{a}$ to $\mathbf{b}$ with $\mathbf{b}(i) = \mathbf{a}(di)$ for all integer $i$, where $\mathbf{b}$ is a binary sequence of period $N$. The $d$-decimation is denoted by $D_d(\cdot)$.
\end{definition}

\begin{definition} \label{def : nega-cyclic shift}
  Let $\mathbf{a}$ be a binary sequence of period $N$. The {\em nega-cyclic shift $\tau$} or {\em odd shift $\tau$}, which is denoted by $\widehat{L}^\tau(\cdot)$, is defined as 
  $$
  \begin{aligned}
    \widehat{L}^{\tau}(\mathbf{a}) =  [ & \mathbf{a}(\tau),\mathbf{a}(\tau + 1), \dots, \mathbf{a}(N-1),\\
    & \mathbf{a}(0)+1,\dots, \mathbf{a}(\tau-1)+1],
  \end{aligned}
  $$
where $0 \le \tau < N$.
\end{definition}

By definition, it is routine to check that the cyclic shift, decimation, and negation operations all preserve the PACF value distribution. Likewise, the nega-cyclic shift, and negation operations preserve the OACF value distribution. Notice that decimation is not OACF-preserving. To explain this, we give the following simple example. 
\begin{example}
Define $\mathbf{s_1}$ and $\mathbf{s_2}$ as:
$$
\begin{aligned}
\mathbf{s}_1 &= [1, 1, 1, 0, 1, 0, 0, 0, 1, 1 ] \\
\mathbf{s}_2 &= [1, 0, 0, 1, 1, 0, 1, 1, 1, 0 ],
\end{aligned}
$$

where $\mathbf{s}_2$ is the $3$-decimation version of $s_1$. The OACF values of $s_1$ are 
$$[10, 0, -2, -4, 2, 0, -2, 4, 2, 0 ],$$ the OACF values of $s_2$ are 
$$[10, 0, -6, 0, 6, 0, -6, 0, 6, 0 ].$$
\end{example}

We now define the nega-decimation.
\begin{definition}\label{def: new decimation 1}
  Let $\mathbf{s}$ be a periodic binary sequence with period $N$. Define $\mathbf{u} = s||(s \oplus 1)$. Let $D_d(\mathbf{u}) = [\mathbf{u}(0),\mathbf{u}(1),\dots,\mathbf{u}(2N-1)]$ for some positive integer $d$ co-prime with $2N$. The nega $d$-decimation on $\mathbf{s}$ induces a binary sequence $\mathbf{s'}$ with period $[\mathbf{u}(0),\mathbf{u}(1),\dots,\mathbf{u}(N-1)]$, and called {\em nega-decimation}, denoted by $\tilde{D}(s)_d = \mathbf{s'}$.
\end{definition}

By Definition \ref{def: new decimation 1}, we have the following equation:
\begin{equation}\label{equ: new decimation 2}
    \tilde{D}_d(\mathbf{s})(\tau) = \mathbf{s}(d\tau) + \left \lfloor \frac{d\tau}{N} \right \rfloor,
\end{equation}
where $\tau = 0,1,\dots N-1$. Note that Eq. \eqref{equ: new decimation 2} in fact defines the nega-decimation, and is exactly the counterpart of the normal decimation in Definition \ref{def : decimation}. We gave the following example to demonstrate the new decimation operation.
\begin{example}\label{example : Hall 31}
   Let $\mathbf{s}$ be a binary sequence of period 31.
$$
\begin{aligned}
  \mathbf{s} = &  [ 0, 1, 1, 1, 1, 0, 1, 0, 1, 0, 0, 0, 1, 0, 0, 1, 1, 1, 0, 0, 0, 0, 0, 1, 1, 0, 0, 1, 0, 1, 1 ].
\end{aligned}
$$
Then we have 

$$
\begin{aligned}
  \mathbf{u} = \mathbf{s||(s \oplus 1)} = & [ 0, 1, 1, 1, 1, 0, 1, 0, 1, 0, 0, 0, 1, 0, 0, 1, 1, 1, 0, 0, 0, 0, 0, 1, 1, 0, 0, 1, 0, 1, 1, \\
& 1, 0, 0, 0, 0, 1, 0, 1, 0, 1, 1, 1, 0, 1, 1, 0, 0, 0, 1, 1, 1, 1, 1, 0, 0, 1, 1, 0, 1, 0, 0].
\end{aligned}
$$
Applying 3-decimation on $\mathbf{u}$, we have
$$
\begin{aligned}
    D_3(\mathbf{u}) = &[ 0, 1, 1, 0, 1, 1, 0, 0, 1, 1, 1, 0, 1, 0, 1, 1, 0, 1, 0, 1, 0, 1, 1, 0, 0, 0, 1, 0, 0, 0, 0,  \\
    & 1, 0, 0, 1, 0, 0, 1, 1, 0, 0, 0, 1, 0, 1, 0, 0, 1, 0, 1, 0,  
1, 0, 0, 1, 1, 1, 0, 1, 1, 1, 1 ].
\end{aligned}
$$
We truncate the first half to obtain $\mathbf{s'} = \tilde{D}_3(\mathbf{s})$:
$$
\mathbf{s'} =  [ 0, 1, 1, 0, 1, 1, 0, 0, 1, 1, 1, 0, 1, 0, 1, 1, 0, 1, 0, 1, 0, 1, 1, 0, 0, 0, 1, 0, 0, 0, 0 ] .
$$
The OACF values of $\mathbf{s}$ are 
$$
[31, 1, -1, -7, -1, -7, -1, 5, -1, 1, -9, 5, 3, 5, 7, 9, -9, -7, -5, -3, -5, 9, -1, 1, -5, 1, 7, 1, 7, 1, -1 ],
$$
and the OACF values of $\mathbf{s'}$ are
$$
[ 31, -7, -1, 1, 3, 9, -5, 9, -5, 1, -1, 1, 7, 1, -5, -7, 7, 5, -1, -7, -1, 1,
-1, 5, -9, 5, -9, -3, -1, 1, 7 ].
$$
Therefore, the OACF distributions of $\mathbf{s}$ and $\mathbf{s'}$ are both
$$
\{ *  (-9)^2, (-7)^3, (-5)^3, -3, (-1)^6, 1^6, 3, 5^3, 7^3, 9^2, 31 * \},
$$
where the exponent denotes the multiplicity in the multiset.
\end{example}

\begin{theorem}\label{thm: nega-decimation}
  The nega-decimation operation is OACF-preserving for periodic binary sequences, and it cannot be obtained by any combination of negation and nega-cyclic shift in general. 
\end{theorem}

To prove this theorem, we need to prove several lemmas. We first give the following property of PACF distribution. 

\begin{lemma}\label{lemma : PACF symmerties}
  Let $\mathbf{u}$ be a binary sequence of period $2N$. If the sequence $\mathbf{u}$ is of the form $\mathbf{u} = \mathbf{s}||(\mathbf{s}\oplus 1)$ for a certain binary sequence $\mathbf{s}$ of length N, then $\mathcal{R}_{\mathbf{u}}(\tau + N) = - \mathcal{R}_{\mathbf{u}}(\tau) $ for all $\tau = 0,1,\dots N-1$.
\end{lemma}
\begin{proof}
  By Eq. \eqref{def: periodic correlation}, we have
  $$
  \begin{aligned}
    \mathcal{{R}}_{\mathbf{u}}(\tau + N) & = \sum_{i=0}^{2N-1}(-1)^{\mathbf{u}(i) - \mathbf{u}(i + N + \tau)} \\
    &= \sum_{i=0}^{2N-1}(-1)^{\mathbf{u}(i) - \mathbf{u}(i + \tau) - 1}\\
    &= -\sum_{i=0}^{2N-1}(-1)^{\mathbf{u}(i) - \mathbf{u}(i +\tau)} \\
    &= -\mathcal{{R}}_{\mathbf{u}}(\tau).
  \end{aligned}
  $$
  The result then follows.
\end{proof}

The OACF distribution of a periodic sequence has the following property.
\begin{lemma}\label{lemma : OACF symmerties}
  Let $\mathbf{a}$ be a binary sequence of period $N$, then $\mathcal{\widehat{R}}_{\mathbf{a}}(\tau) = -\mathcal{\widehat{R}}_{\mathbf{a}}(N-\tau)$ for all $1 \le \tau\le N-1$.
\end{lemma}
\begin{proof}
  For $1\le \tau \le N-1$, by Eq. \eqref{def: odd periodic correlation}, we have
  $$
  \begin{aligned} 
    \mathcal{\widehat{R}}_{\mathbf{a}}(\tau) &= \sum_{i=0}^{N-1-\tau}(-1)^{\mathbf{a}(i) - \mathbf{a}(i+\tau)} + \sum_{i=N-\tau}^{N-1}(-1)^{\mathbf{a}(i) - \mathbf{a}(i+\tau)-1} \\
    &= \sum_{i=0}^{N-1-\tau}(-1)^{\mathbf{a}(i) - \mathbf{a}(i+\tau)} - \sum_{i=0}^{\tau-1}(-1)^{\mathbf{a}(i-\tau) - \mathbf{a}(i)}.
  \end{aligned}
  $$
 Since $1\le N-\tau\le N-1$, it then follows that 
  $$
  \begin{aligned}
    \mathcal{\widehat{R}}_{\mathbf{a}}(N-\tau) &= \sum_{i=0}^{\tau-1}(-1)^{\mathbf{a}(i) - \mathbf{a}(i-\tau)} - \sum_{i=0}^{N-\tau-1}(-1)^{\mathbf{a}(i+\tau) - \mathbf{a}(i)}\\
    &= -\mathcal{\widehat{R}}_{\mathbf{a}}(\tau)
  \end{aligned}
  $$
The proof is then completed.
\end{proof}

By Lemma \ref{lemma : PACF symmerties}, Lemma \ref{lemma : OACF symmerties}, and Eq. \eqref{equ : Parker transformation}, if two sequences $\mathbf{u}=\mathbf{s}||(\mathbf{s}\oplus 1)$, $\mathbf{u'} = \mathbf{s'}||(\mathbf{s'}\oplus 1)$ have the same PACF value distribution, then $\mathbf{s}$ and $\mathbf{s'}$ have the same OACF value distribution. This result is stated in the following lemma, which will play a key role in proving Theorem \ref{thm: nega-decimation}.

\begin{lemma} \label{lemma: OACF-preserving based on Parker's transformation}
  Let $\mathbf{u} = \mathbf{s}||(\mathbf{s}\oplus 1)$ for a binary sequence $\mathbf{s}$ with period $N$. Let $\mathbf{u'} = \mathbf{s'} ||(\mathbf{s'} \oplus 1)$ for a binary sequence $\mathbf{s'}$ of period $N$. The distribution of OACF values of $\mathbf{s'}$ are identical to that of $\mathbf{s}$ if the distribution of PACF values of $\mathbf{u}$ and $\mathbf{u'}$ are identical.
\end{lemma}

\begin{proof}
  Define mutisets $S_{\mathbf{u}} = \{ *  \mathcal{R}_{\mathbf{u}}(i) : i =0,1,\dots, 2N-1 *  \}$ and $S_{\mathbf{u'}} = \{  * \mathcal{R}_{\mathbf{u'}}(i) : i =0,1,\dots, 2N-1 *  \}$, respectively.

  By Lemma \ref{lemma : PACF symmerties} and Eq. \eqref{equ : Parker transformation}, we have 

  \begin{eqnarray*}
    S_{\mathbf{u}}
    & = & \{ * 2N,-2N * \} \cup \{ * \mathcal{R}_{\mathbf{u}}(i) : i=1,\dots, N-1   *\} \\
    &  & \cup \{  * -\mathcal{R}_{\mathbf{u}}(i) : i=1,\dots, N-1 *  \}  \\
    & = & \{ * 2N,-2N *  \} \cup \{ *  2\mathcal{\widehat{R}}_{\mathbf{s}}(i) : i=1,\dots, N-1  * \}  \\
    &  & \cup \{  * -2\mathcal{\widehat{R}}_{\mathbf{s}}(i) : i=1,\dots, N-1 *   \},
  \end{eqnarray*}
  and 
  \begin{eqnarray*}
    S_{\mathbf{u'}} & =  & \{ * 2N,-2N * \} \cup \{ * 2\mathcal{\widehat{R}}_{\mathbf{s'}}(i) : i=1,\dots, N-1 * \} \\
    & & \cup \{ *  -2\mathcal{\widehat{R}}_{\mathbf{s'}}(i) : i=1,\dots, N-1   *\}.
  \end{eqnarray*}
  When $N$ is odd, by Lemma \ref{lemma : OACF symmerties}, we have
  \begin{eqnarray*}
     \lefteqn{\{ *  2\mathcal{\widehat{R}}_{\mathbf{s}}(i) : i=1,\dots, N-1   *\}} \\
    & = &\{ * 2\mathcal{\widehat{R}}_{\mathbf{s}}(i) : i=1,\dots, (N-1)/2  * \} \cup \\
      & & \{  * -2\mathcal{\widehat{R}}_{\mathbf{s}}(i) : i=1,\dots, (N-1)/2   *\},
  \end{eqnarray*}
  and 
  \begin{eqnarray*}
    \lefteqn{\{ *  2\mathcal{\widehat{R}}_{\mathbf{s'}}(i) : i=1,\dots, N-1   *\}}\\
    & = & \{  * 2\mathcal{\widehat{R}}_{\mathbf{s'}}(i) : i=1,\dots, (N-1)/2   *\} \cup \\
     & & \{ *  -2\mathcal{\widehat{R}}_{\mathbf{s'}}(i) : i=1,\dots, (N-1)/2   *\}.
  \end{eqnarray*}
  It follows that 
  $$\{ *\mathcal{R}_{u}(i) : i=1,\dots, N-1 * \} = \{ *  -\mathcal{R}_{u}(i) : i=1,\dots, N-1 *\} $$ 
  and 
  $$\{* \mathcal{R}_{u'}(i) : i=1,\dots, N-1*\} = \{*-\mathcal{R}_{u'}(i) : i=1,\dots, N-1 *\}.$$ 
  Therefore, we have 
  $$
  S_{\mathbf{u}} = \{*2N,-2N*\} \cup \{*\mathcal{R}_{u}(i)^2 : i= 1,\dots, N-1*\},
  $$ 
  and 
  $$
  S_{\mathbf{u'}} = \{*2N,-2N*\} \cup \{*\mathcal{R}_{u'}(i)^2 : i= 1,\dots, N-1*\},
  $$
  where the exponent 2 is the multiplicity. 
  The relation $S_{\mathbf{u}} = S_{\mathbf{u'}}$ forces that $\{* \mathcal{R}_{u}(i) : i= 1,\dots N-1* \} = \{* \mathcal{R}_{u'}(i) : i= 1,\dots N-1 *\}$. From Eq. \eqref{equ : Parker transformation} and $\mathcal{R}_{\mathbf{u}}(0) = \mathcal{R}_{\mathbf{u'}}(0) =2N$, it follows that the OACF values of $\mathbf{s}$ are identical to the OACF values of $\mathbf{s'}$.
 
  When $N$ is even, $\mathcal{\widehat{R}}_{\mathbf{u}}(N/2) = - \mathcal{\widehat{R}}_{\mathbf{u}}(N/2) = 0$. By Lemma \ref{lemma : OACF symmerties}, we have 
  \begin{eqnarray*}
     \lefteqn{\{*2\mathcal{\widehat{R}}_{\mathbf{s}}(i) : i=1,\dots, N-1*\}}\\
    & = &\{* 0 *\} \cup \{* 2\mathcal{\widehat{R}}_{\mathbf{s}}(i) : i=1,\dots, N/2 -1* \} \cup \\
    &  & \{*-2\mathcal{\widehat{R}}_{\mathbf{s}}(i) : i=1,\dots, N/2 -1 *\},
  \end{eqnarray*}
 
  and 
  \begin{eqnarray*}
    \lefteqn{\{ * 2\mathcal{\widehat{R}}_{\mathbf{s'}}(i) : i=1,\dots, N-1 *\}}\\
    & = &\{*0*\} \cup \{*2\mathcal{\widehat{R}}_{\mathbf{s'}}(i) : i=1,\dots, N/2 -1 *\} \cup \\
    & & \{*-2\mathcal{\widehat{R}}_{\mathbf{s'}}(i) : i=1,\dots, N/2 -1 *\}.
  \end{eqnarray*}

  This implies that 
  $$\{*\mathcal{R}_{u}(i) : i=1,\dots, N-1*\} = \{*-\mathcal{R}_{u}(i) : i=1,\dots, N-1*\}, $$ 
  and 
  $$\{*\mathcal{R}_{u'}(i) : i=1,\dots, N-1*\} = \{*-\mathcal{R}_{u'}(i) : i=1,\dots, N-1*\}.$$ 
  The rest of the proof is the same as the case that $N$ is odd.  
\end{proof}

Now we are ready to give a proof of Theorem \ref{thm: nega-decimation}. The proof of Theorem \ref{thm: nega-decimation} is based on the PACF distribution of $\mathbf{u}$ and Eq. \eqref{equ : Parker transformation}.

\begin{proof}[Proof of Theorem \ref{thm: nega-decimation}]
Let $\mathbf{u} = \mathbf{s}||(s\oplus 1)$, where $\mathbf{s}$ is a periodic binary sequence with period $N$. For a positive integer $d$ that coprime with $2N$, denote $D_d(\mathbf{u})$ by $\mathbf{u'}$, we have the following observation:
$$\mathbf{u'}(i+N) = \mathbf{u}(di + dN) = \mathbf{u}(di) + 1 = \mathbf{u'}+1\text{, } 0 \le i < N.$$

Thus, the sequence $\mathbf{u'}$ must has the form $\mathbf{s'}||(\mathbf{s'}\oplus 1)$ for a $\mathbf{s'}$ with period $N$.
By Lemma \ref{lemma: OACF-preserving based on Parker's transformation}, we know that nega-$d$-decimation is OACF preserving.

It remains to show that nega-decimation cannot be obtained by the previously known operations in general. In Example \ref{example : Hall 31}, by computer search, the sequence $\mathbf{s'}$ cannot be obtained from any combination of nega-cyclic shift, negation. Thus the nega-decimation is indeed a new operation and the proof of Theorem \ref{thm: nega-decimation} is complete.
\end{proof}

\begin{remark}
  As far as we are concerned, there are three known OACF-preserving operations: negation, nega-cyclic shift and nega-decimation. 
\end{remark}

\section{Classification of Parker's 16 cases} \label{Sec: section 3}

In \cite{parker2001even}, Parker constructed two classes of odd-optimal binary sequences of period $2p$ and further conjectured 16 classes of binary sequences' OACF values of period $4p$ with $p$ prime. These conjectures were recently confirmed by Li and Yang \cite{li2019three}. In this section, we first review Parker's 16 classes of sequences and then classify them into 8 ones by the newly defined nega-decimation operation.

\subsection{Parker's sequences}
Let $p$ be a prime of the form $p=x^2 + 4y^2 = 4f + 1$, where $x, y, f$ are integers. Fix $\alpha$ as a generator of the finite field of order $p$, denoted as $GF(p)$, and the cyclotomic classes of order 4 with respect to $GF(p)$ are $\{D_0,D_1,D_2,D_3\}$. Define $S = \{G \times \{0\} | G \subseteq \mathbb{Z}_8\} \cup \{\{n\} \times A_n | n \in \mathbb{Z}_8\}$ where $A_n = \cup_{k \in  I_n}D_k $ with $  I_n \subseteq \mathbb{Z}_{4}$. 
By the Chinese Remainder Theorem (CRT), $\mathbb{Z}_8\times \mathbb{Z}_p$ is isomorphic to $\mathbb{Z}_{8p}$. Let $\eta$ be an isomorphism from $\mathbb{Z}_8\times \mathbb{Z}_p$ to $\mathbb{Z}_{8p}$ and $N=4p$. The characteristic sequence $\mathbf{u}$ of $\eta(S)$ satisfies $\mathbf{u}(i) = \mathbf{u}(i + N) + 1$ for $i = 0,1,\dots N-1$ if we require that $A_{n + 4} = \cup_{t \notin  I_n} D_t$ and $j + 4 \in (\notin) G$ for all $j \notin (\in) G$, for $j,n=0,1,2,3$ \cite{parker2001even}. Thus, $\mathbf{u}$ is completely described by $G', A_0, A_1, A_2, A_3$, where $G' = \{g \in G, g \le 4\}$. The OACF values of Parker's 16 classes of sequences were confirmed by Li and Yang \cite{li2019three} via an interleaving argument, and the results are listed in Tables \ref{tab:conjecture 1}, \ref{tab:conjecture 2}, \ref{tab:conjecture 3}. Each class of sequences is denoted by $\mathbf{s_i},$ for $ i = 1,2,\dots, 16$, and the support of $\mathbf{s}_i||(\mathbf{s}_i\oplus 1)$ is determined by $G'$ and $\gamma $, where $\gamma = A_0, A_1, A_2, A_3$. The sets $C_i$'s for $i = 1, 2, \ldots, 6$ are defined as follows : 

$$
C_1 = D_0 \cup D_1,\text{ } C_2 = D_0 \cup D_2, \text{ } C_3 = D_0 \cup D_3,
$$
$$
C_4 = D_1 \cup D_2, \text{ }C_5 = D_1 \cup D_3, \text{ } C_6 = D_2 \cup D_3.
$$

\begin{table}[htbp]
  \caption{Construction 1}
   \centering
   \resizebox{0.7\linewidth}{!}{
   \begin{tabular}{c|c|c|c|c}
     \hline
        Notation  &    $G'$     & $\gamma$     & OACF values  & $f$\\
            \hline
          $\mathbf{s}_1$ &  $\{2\}$     & $C_3,C_4,C_1,C_1$ & $\{0,\pm 2, \pm 4, \pm(2x+4y)\}$    & even  \\
          $\mathbf{s}_2$ &  $\{0,1,2\}$ & $C_4,C_3,C_1,C_1$ & $\{0,\pm 2, \pm 4, \pm(2x-4y)\}$    & even  \\
          $\mathbf{s}_3$ &  $\{3\}$     & $C_6,C_1,C_4,C_4$ & $\{0,\pm 2, \pm 4, \pm(2x-4y)\}$    & even  \\
          $\mathbf{s}_4$ &  $\{0,1,3\}$ & $C_1,C_6,C_4,C_4$ & $\{0,\pm 2, \pm 4, \pm(2x+4y)\}$    & even  \\
     \hline
   \end{tabular}}
   \label{tab:conjecture 1}
 \end{table}

 \begin{table}[htbp]
  \caption{Construction 2}
   \centering
   \resizebox{0.7\linewidth}{!}{
   \begin{tabular}{c|c|c|c|c}
     \hline
        Notation  &    $G'$     & $\gamma$     & OACF values  & $f$\\
            \hline
          $\mathbf{s}_5$ &  $\{0\}$     & $C_3,C_4,C_1,C_1$ & $\{0,\pm 2, \pm 4, \pm(2x+4y)\}$    & odd  \\
          $\mathbf{s}_6$ &  $\{1\}$ & $C_4,C_3,C_1,C_1$ & $\{0,\pm 2, \pm 4, \pm(2x-4y)\}$    & odd  \\
          $\mathbf{s}_7$ &  $\{0,2,3\}$     & $C_6,C_1,C_4,C_4$ & $\{0,\pm 2, \pm 4, \pm(2x-4y)\}$    & odd  \\
          $\mathbf{s}_8$ &  $\{1,2,3\}$ & $C_1,C_6,C_4,C_4$ & $\{0,\pm 2, \pm 4, \pm(2x+4y)\}$    & odd  \\
     \hline
   \end{tabular}}
   \label{tab:conjecture 2}
 \end{table}

 \begin{table}[htbp]
  \caption{Construction 3}
   \centering
   \resizebox{0.7\linewidth}{!}{
   \begin{tabular}{c|c|c|c|c}
     \hline
        Notation  &    $G'$     & $\gamma$     & OACF values  & $f$\\
            \hline
          $\mathbf{s}_9$ &  $\{0\}$     & $C_5,C_2,C_1,C_1$ & $\{0,\pm 2, \pm 2y, \pm 4y, \pm (4y+8)\}$    & odd  \\
          $\mathbf{s}_{10}$ &  $\{0\}$ & $C_2,C_5,C_1,C_1$ & $\{0,\pm 2, \pm 2y, \pm 4y, \pm (4y-8)\}$    & odd  \\
          $\mathbf{s}_{11}$ &  $\{0\}$     & $C_5,C_2,C_4,C_4$ & $\{0,\pm 2, \pm 2y, \pm 4y, \pm (4y-8)\}$    & odd  \\
          $\mathbf{s}_{12}$ &  $\{0\}$ & $C_2,C_5,C_4,C_4$ & $\{0,\pm 2, \pm 2y, \pm 4y, \pm (4y+8)\}$    & odd  \\
          $\mathbf{s}_{13}$ &  $\{0,3\}$     & $C_1,C_2,C_2,C_1$ & $\{0,\pm 2, \pm 2y, \pm 4y, \pm (4y+8)\}$    & odd  \\
          $\mathbf{s}_{14}$ &  $\{0,3\}$ & $C_1,C_5,C_5,C_1$ & $\{0,\pm 2, \pm 2y, \pm 4y, \pm (4y-8)\}$    & odd  \\
          $\mathbf{s}_{15}$ &  $\{0,3\}$     & $C_4,C_2,C_2,C_4$ & $\{0,\pm 2, \pm 2y, \pm 4y, \pm (4y-8)\}$    & odd  \\
          $\mathbf{s}_{16}$ &  $\{0,3\}$ & $C_4,C_5,C_5,C_4$ & $\{0,\pm 2, \pm 2y, \pm 4y, \pm (4y+8)\}$    & odd  \\
     \hline
   \end{tabular}}
   \label{tab:conjecture 3}
 \end{table}

\subsection{Classification of Parker's sequences}

We give a proof of the case that $\mathbf{s_4}$ can be obtained by $\mathbf{s_1}$ via a combination of the negation and nega-decimation. The other 14 Parker's sequences can be handled in a similar way, and the proof is thus omitted. 
\begin{theorem} \label{thm : s_1 equiva s_4}
  Let $\mathbf{s}_1,\mathbf{s}_4$ be the sequences defined in Table \ref{tab:conjecture 1}. Then the sequence $\mathbf{s}_4$ can be obtained from $\mathbf{s}_1$ via a combination of negation, and nega-decimation. More precisely, $\mathbf{s}_4 = \tilde{D}_{\phi^{-1}((1,\alpha^3))}(\mathbf{s}_1\oplus 1)$, where $\alpha$ is the fixed generator in $GF(p)$ and $\phi$ is an isomorphism from $\mathbb{Z}_{8p}$ to $\mathbb{Z}_8 \times \mathbb{Z}_p$ that maps $x$ to $(x \mod 8, x \mod p)$.
\end{theorem}
\begin{proof}
  Define $\mathbf{u}_1 = \mathbf{s}_1 || (\mathbf{s}_1 \oplus 1)$ and $\mathbf{u}_4 = \mathbf{s}_4 || (\mathbf{s}_4 \oplus 1) $.  Let $S_{\mathbf{u}_1}$, $S_{\mathbf{u}_4}$ be the supports of $\mathbf{u}_1,\mathbf{u}_2$ respectively, and
  $$
  \begin{aligned}
    \phi(S_{\mathbf{u}_1}) = \text{ } & C_{\mathbf{u}_1}  \\
                      = \text{ } & \{(2,0),(4,0), (5,0),(7,0)\}  \\
                     & \cup (0,C_3)\cup (1,C_4) \cup (2,C_1) \cup (3,C_1) \\
                     & \cup (4,C_4)\cup (5,C_3) \cup (6,C_6) \cup (7,C_6), \\
    \phi(S_{\mathbf{u}_4}) = \text{ } & C_{\mathbf{u}_4} \\
                    = \text{ } &  \{(0,0),(1,0),(3,0),(6,0)\}  \\
                    & \cup   (0,C_1)\cup (1,C_6) \cup (2,C_4) \cup (3,C_4) \\
                    & \cup   (4,C_6)\cup (5,C_1) \cup (6,C_3) \cup (7,C_3). \\
  \end{aligned}
  $$
  Firstly, we apply negation on $\mathbf{u}_1$. Then the support of $\mathbf{u}_1\oplus 1$ that is denoted by $S_{\mathbf{u}_1\oplus 1}$ satisfies 
  \begin{eqnarray*}
    \lefteqn{\phi(S_{\mathbf{u}_1\oplus 1})} \\
      & = & \mathbb{Z}_8 \times \mathbb{Z}_p  \setminus C_{\mathbf{u}_1} \\
      & = & \{(0,0),(1,0),(3,0),(6,0)\}  \\
      & & \cup (0,C_4)\cup (1,C_3) \cup (2,C_6) \cup (3,C_6) \\
      & & \cup (4,C_3)\cup (5,C_4) \cup (6,C_1) \cup (7,C_1).
  \end{eqnarray*}
 We then apply the $d$-decimation on $\mathbf{u}_1\oplus 1$, where $\phi(d) = (1,\alpha^3)$. The support of $D_d(\mathbf{u}_1\oplus 1)$ that is denoted by $S_{D_d(\mathbf{u}_1\oplus 1)}$ satisfies
  \begin{eqnarray*}
    \lefteqn{\phi(S_{D_d(\mathbf{u}_1\oplus 1)})} \\
    & = & (1,\alpha^3) \cdot \phi(S_{\mathbf{u}_1\oplus 1}) \\
    & = & \{(1\times0,\alpha^3 \cdot 0),(1\times1,\alpha^3 \cdot 0), \\
    &  &  (1\times3,\alpha^3 \cdot 0),(1\times6,\alpha^3 \cdot 0)\} \\
    &  & \cup  (1 \times 0,\alpha^3 \cdot C_4)\cup (1\times 1,\alpha^3 \cdot C_3) \\
    & & \cup  (1\times 2,\alpha^3 \cdot C_6) \cup (1\times 3,\alpha^3 \cdot C_6) \\
    & & \cup (1\times 4,\alpha^3 \times C_3) \cup (1\times 5,\alpha^3 \cdot C_4) \\
    & & \cup  (1\times 6,\alpha^3 \cdot C_1) \cup (1 \times 7,\alpha^3 \cdot C_1) \\
    & = &  \{(0,0),(1,0),(3,0),(6,0)\} \\
    &  & \cup  (0,C_1)\cup (1,C_6) \cup (2,C_4) \cup (3,C_4) \\
    &  & \cup  (4,C_6)\cup (5,C_1) \cup (6,C_3) \cup (7,C_3) \\
    & = & C_{\mathbf{u}_4}.
  \end{eqnarray*}

  This implies that $S_{D_d(\mathbf{u}_1\oplus 1)} = S_{\mathbf{u}_4}$ since $\phi$ is bijective. Therefore, $\mathbf{s}_4 = \tilde{D}_{\phi^{-1}((1,\alpha^3))}(\mathbf{s}_1\oplus 1)$. This means that the sequence $\mathbf{s}_4$ can be obtained from $\mathbf{s}_1$ via a combination of negation and nega-decimation.
\end{proof}

\begin{remark}
  The interleaving-based construction in \cite{li2019three} can be represented by the subsets in $\mathbb{Z}_8\times\mathbb{Z}_p$ via an isomorphism. The subset has the same format as $C_{\mathbf{u}_1}$ in the proof of Theorem \ref{thm : s_1 equiva s_4}. Thus, we can classify the interleaving based construction by the method in Theorem \ref{thm : s_1 equiva s_4}.
\end{remark}

By applying the same method in Theorem \ref{thm : s_1 equiva s_4}, we can classify the 16 classes of binary sequences by Parker into 8 classes in Table \ref{tab : Classification of Parker's Sequences}, where $\phi$ and $\alpha$ are defined as in Theorem \ref{thm : s_1 equiva s_4}.

\begin{table}[htbp]
  \caption{Classification of Parker's 16 cases}
   \centering
  \begin{tabular}{c|c|c}
    \hline
  Class Index & Sequences & Relation\\
  \hline
  1 & $\mathbf{s_1,s_4}$ & $\mathbf{s}_4 = \tilde{D}_{\phi^{-1}((1,\alpha^3))}(\mathbf{s}_1\oplus 1)$ \\
  2 & $\mathbf{s_2,s_3}$ & $\mathbf{s}_3 = \tilde{D}_{\phi^{-1}((1,\alpha^3))}(\mathbf{s}_2\oplus 1)$ \\
  3 & $\mathbf{s_5,s_8}$ & $\mathbf{s}_8 = \tilde{D}_{\phi^{-1}((1,\alpha^3))}(\mathbf{s}_5\oplus 1)$ \\
  4 & $\mathbf{s_6,s_7}$ & $\mathbf{s}_7 = \tilde{D}_{\phi^{-1}((1,\alpha^3))}(\mathbf{s}_6\oplus 1)$ \\
  5 & $\mathbf{s_9,s_{12}}$ & $\mathbf{s}_{12} = \tilde{D}_{\phi^{-1}((1,\alpha^3))}(\mathbf{s}_9)$ \\
  6 & $\mathbf{s_{10},s_{11}}$ & $\mathbf{s}_{11} = \tilde{D}_{\phi^{-1}((1,\alpha))}(\mathbf{s}_{10})$ \\
  7 & $\mathbf{s_{13},s_{16}}$ & $\mathbf{s}_{16} = \tilde{D}_{\phi^{-1}((1,\alpha))}(\mathbf{s}_{13})$ \\
  8 & $\mathbf{s_{14},s_{15}}$ & $\mathbf{s}_{15} = \tilde{D}_{\phi^{-1}((1,\alpha))}(\mathbf{s}_{14})$ \\
  \hline
  \end{tabular}
  \label{tab : Classification of Parker's Sequences}
  \end{table}

\section{Conclusion} \label{Sec: Conclusion}
In this paper, we defined a new OACF-preserving operation called nega-decimation. The nega-decimation operation enables us to classify the 16 Parker's classes of sequences into 8 ones. The interleaving sequences defined by Li and Yang \cite{li2019three} can also be classified by the new operation. Moreover, we may possibly be able to fully classify all constructions of binary sequences based on Parker's transformation. This may constitute one of future directions.
\bibliographystyle{ieeetr}  

\begin{thebibliography}{1}
\bibitem{cai2009binary}
Y.~Cai and C.~Ding, ``Binary sequences with optimal autocorrelation,'' {\em
  Theoretical Computer Science}, vol.~410, no.~24-25, pp.~2316--2322, 2009.

\bibitem{yang2015binary}
Y.~Yang, X.~Tang, and Z.~Zhou, ``Binary sequences with optimal odd periodic
  autocorrelation,'' in {\em 2015 IEEE International Symposium on Information
  Theory (ISIT)}, (Hong Kong, China), pp.~1551--1554, IEEE, June 2015.

\bibitem{li2019three}
C.~Li and Y.~Yang, ``On three conjectures of binary sequences with low
  odd-periodic autocorrelation,'' {\em Cryptography and Communications},
  pp.~1--16, 2019.
\newblock doi: 10.1007/s12095-019-00393-3.

\bibitem{yang2017generic}
Y.~Yang and X.~Tang, ``Generic construction of binary sequences of period $2 n
  $ with optimal odd correlation magnitude based on quaternary sequences of odd
  period $ n$,'' {\em IEEE Transactions on Information Theory}, vol.~64, no.~1,
  pp.~384--392, 2017.

\bibitem{pott2008two}
A.~Pott, ``Two applications of relative difference sets: Difference triangles
  and negaperiodic autocorrelation functions,'' {\em Discrete mathematics},
  vol.~308, no.~13, pp.~2854--2861, 2008.

\bibitem{parker2001even}
M.~G. Parker, ``Even length binary sequence families with low negaperiodic
  autocorrelation,'' in {\em 14th International Symposium on Applied Algebra,
  Algebraic Algorithms, and Error-Correcting Codes (AAECC)}, (Melbourne,
  Australia), pp.~200--209, Springer, Nov. 2001.

\bibitem{wang2019oacf}
G.~Wang and Q.~Wang, ``{An OACF-Preserving operation based on Parker’s
  Transformation},'' in {\em 2019 Ninth International Workshop on Signal Design
  and its Applications in Communications (IWSDA)}, (Dongguan, China), pp.~1--5,
  IEEE, Oct. 2019.
\end{thebibliography}

\end{document}